\documentclass[letterpaper, 10 pt,conference]{ieeeconf}  
\usepackage{amssymb}
\usepackage{amsmath,mathrsfs}
\usepackage{color}
\usepackage{graphicx}
\usepackage{mathtools}
\usepackage{lineno,hyperref}
\usepackage{cleveref}
\usepackage{bm}
\usepackage{tikz}
\usepackage{algorithm}
\usepackage{algorithmicx}
\usepackage{algpseudocode}
\modulolinenumbers[5]
\newtheorem{theorem}{Theorem}[section]
\newtheorem{remark}{Remark}[section]

\newtheorem{proposition}[theorem]{Proposition}
\newtheorem{problem}{Problem}[section]

\usepackage{multirow}
\usepackage{color}
\usepackage{cite}

\providecommand{\keywords}[1]{\textbf{\textit{Index terms---}}}

\IEEEoverridecommandlockouts                              
\title{\LARGE \bf Non-causal regularized least-squares for continuous-time system identification with band-limited input excitations}

\author{Rodrigo A. \text{Gonz\'alez}, Cristian R. Rojas, and H{\aa}kan Hjalmarsson 
\thanks{This work was supported by the Swedish Research Council under contract number 2016-06079 (NewLEADS). The authors are with the Division of Decision and Control Systems, KTH Royal Institute of Technology, 10044 Stockholm, Sweden (e-mails: grodrigo@kth.se; crro@kth.se; hjalmars@kth.se).}%
}
\begin{document}

\maketitle

\begin{abstract}
	In continuous-time system identification, the intersample behavior of the input signal is known to play a crucial role in the performance of estimation methods. One common input behavior assumption is that the spectrum of the input is band-limited. The sinc interpolation property of these input signals yields equivalent discrete-time representations that are non-causal. This observation, often overlooked in the literature, is exploited in this work to study non-parametric frequency response estimators of linear continuous-time systems. We study the properties of non-causal least-square estimators for continuous-time system identification, and propose a kernel-based non-causal regularized least-squares approach for estimating the band-limited equivalent impulse response. The proposed methods are tested via extensive numerical simulations.
\end{abstract}

\section{Introduction}

Continuous-time system identification studies how to obtain continuous-time mathematical models of systems based on sampled input and output data. This field, together with its discrete-time counterpart, has had a deep impact in many areas of science and engineering, and significant pieces of literature have been written on the subject, see, e.g.,~\cite{rao2006identification,garnier2008book}.

In both continuous and discrete-time system identification, methods should be picked according to the assumptions the user makes on the input signal. Three main assumptions can be commonly found: that the input is piecewise constant, piecewise linear, or band-limited, the latter meaning that the power spectrum of the signal is zero above a certain frequency. 
Due to the advantages provided by the Nyquist-Shannon reconstruction theorem, which permits the exact intersample behavior of the signal to be known based on samples, band-limited signals have been studied extensively in signal processing, filter theory, and spectral theory. In system identification, parametric continuous-time system identification under band-limited inputs has been carried out mostly in the frequency domain \cite{pintelon1997frequency}, in which the description of such signals is natural. For the case of continuous-time multisine inputs, a least-squares method in the frequency domain that has been widely used is Levy's method \cite{levy1959complex}, and a time-domain method based on refined instrumental variables has recently been introduced and analyzed in~\cite{gonzalez2020consistent}. 

In this paper we first show that, under band-limited input assumptions, the equivalent discrete-time system is non-causal. A similar observation has been made in, e.g., \cite{feuer1996sampling}, where it is stated that the reconstruction of band-limited signals is a non-causal filtering procedure. However, the implications of this sampling result to the equivalent discrete-time system description seem to have been overlooked in the literature (see, e.g., Eq. (2) of \cite{relan2016recursive}, in which a direct term has been added instead of a fully non-causal discrete-time representation). 
 Band-limited input signals were also used in \cite{rabiner1978fir} for estimating a discrete-time finite-impulse response (FIR) filter. This work was done in the discrete-time domain and the non-causal components that arise were also neglected.

The non-causal discrete-time equivalent system description leads to direct ways to estimate the non-causal impulse response, and therefore the continuous-time frequency response, based on sampled band-limited data. Non-causal system identification has been recently studied in \cite{lu2019identification}, in which a Maximum Likelihood estimator is proposed for symmetric non-causal systems with applications to cross direction modeling of paper machines. In particular, non-causal FIR models have been used for identifying systems in closed-loop in \cite{forssell2000projection, aljanaideh2017closed}.

In summary, the main results of this paper are:
\begin{itemize}
	\item
	We show that the equivalent discrete-time system~arising from a band-limited input signal is non-causal, and we analyze its properties.
	\item
	We propose a least-squares estimator for computing the non-causal discrete-time impulse response, and derive its asymptotic distribution.
	\item
	We present a non-causal regularized least-squares method for estimating the non-causal impulse response that represents the continuous-time system.
	\item
	We illustrate our methods via extensive Monte Carlo simulations.	
\end{itemize}

The rest of this work is organized as follows. In Section \ref{sec:preliminaries} we introduce basic concepts of band-limited signals and their implications on linear systems, and we state the problem we study. In Section \ref{sec:approaches} we present the least-squares approach for estimating the non-causal band-limited equivalent impulse response, and propose non-causal regularization methods for improving its performance. Section \ref{sec:simulations} contains extensive numerical experiments evaluating the algorithms, and we provide concluding remarks in Section \ref{sec:conclusions}. Proofs of the theoretical results can be found in the Appendix.

\section{Preliminaries}
\label{sec:preliminaries}
In this section, we will discuss the topic of sampling band-limited signals from a continuous-time system standpoint. In particular, we recall the concept of a band-limited signal, and introduce the non-causal discrete-time impulse response that is obtained under band-limited assumptions in the input. 

Consider the following system description
\begin{equation}
\label{system}
x(t) = \int_{0}^{\infty} g(\tau) u(t-\tau) \textnormal{d}\tau,
\end{equation}
where $\{u(t)\}$ is a scalar input of the continuous-time, asymptotically stable, linear and time invariant system that has a causal impulse response $\{g(t)\}$, and $\{x(t)\}$ is the output. The frequency response of the system and the continuous-time Fourier transform of the input are given by
\begin{equation}
G(i\omega) = \int_{0}^\infty g(t) e^{-i\omega t}\textnormal{d}t \hspace{0.2cm}\textnormal{and}\hspace{0.2cm} U(i\omega) = \int_{-\infty}^{\infty} u(t) e^{-i\omega t} \textnormal{d}t \notag
\end{equation}
respectively. The key assumption in this work is that the input signal is \textit{band-limited}, that is, $\{u(t)\}$ does not have energy above a certain frequency $\omega_B$. In other words, $U(i\omega) = 0$ for $|\omega| >\omega_B$. If $u(t)$ is sampled every $h$ seconds, where $h<\pi/\omega_B$, then its discrete-time Fourier transform pair is given by
\begin{equation}
U_{\hspace{-0.02cm}h}\hspace{-0.03cm} (\hspace{-0.02cm}e^{i\omega h}\hspace{-0.025cm})\hspace{-0.11cm} =\hspace{-0.08cm}  h\hspace{-0.22cm} \sum_{k=-\infty}^{\infty} \hspace{-0.25cm} u\hspace{-0.02cm}(\hspace{-0.01cm}kh\hspace{-0.01cm}) e^{\hspace{-0.04cm} -i\omega kh} \hspace{-0.24cm} \iff \hspace{-0.17cm}  u\hspace{-0.02cm}(\hspace{-0.01cm}kh\hspace{-0.01cm})\hspace{-0.1cm}  = \hspace{-0.09cm}\hspace{-0.07cm} \int_{\hspace{-0.07cm} -\hspace{-0.02cm}\frac{\pi}{h}}^{\frac{\pi}{h}} \hspace{-0.16cm} U_{\hspace{-0.02cm}h}\hspace{-0.04cm} (\hspace{-0.02cm}e^{i\omega h}\hspace{-0.02cm}) \hspace{-0.04cm}\frac{e^{\hspace{-0.01cm}i\omega kh}}{2\pi} \textnormal{d}\omega. \notag
\end{equation}
These expressions can be exploited so that the discrete-time Fourier transform is written in terms of the continuous-time one, which is known as Poisson's summation formula \cite{boas1972summation}
\begin{equation}
U_h(e^{i\omega h}) = \sum_{n=-\infty}^\infty  U\left(i\omega +i \frac{2\pi n}{h} \right). \notag
\end{equation}
Due to $\{u(t)\}$ being band-limited, this formula indicates that $U_h(e^{i\omega h})=U(i\omega)$ for $|\omega|<\pi/h$. Since $X(i\omega) = G(i\omega) U(i\omega)$, we find that $\{x(t)\}$ is also band-limited, and thus $X_h(e^{i\omega h}) = X(i\omega)$ in the same domain. Using these identities, we can exactly reconstruct a continuous-time band-limited signal based on its samples:
\begin{align}
u(t) &= \frac{1}{2\pi} \int_{-\frac{\pi}{h}}^{\frac{\pi}{h}} U_h(e^{i\omega h}) e^{i\omega t}\textnormal{d}\omega \notag \\
&= \sum_{n=-\infty}^\infty u(nh) \frac{h}{2\pi} \int_{-\frac{\pi}{h}}^{\frac{\pi}{h}} e^{i\omega (t-nh)} \textnormal{d}\omega \notag \\
\label{sincinterpolation}
&= \sum_{n=-\infty}^\infty u(nh) \textnormal{ sinc}\left( \frac{t-nh}{h}\right),
\end{align}
where the sinc function is defined as $\textnormal{sinc}(t):= \sin(\pi t)/(\pi t)$. Replacing this description of $u(t)$ in \eqref{system} and interchanging summation and integration, the system equation can then be rewritten as
\begin{equation}
\label{sampledoutput}
x(t) = \sum_{n=-\infty}^{\infty} u(nh) \int_{0}^{\infty} g(\tau) \textnormal{ sinc}\left( \frac{t-\tau-nh}{h}\right) \textnormal{d}\tau.
\end{equation}
Thus, we have the following result.
\begin{proposition}
	The equivalent discrete-time model of a system whose input is a band-limited signal is described by the impulse response
	\begin{equation}
	\label{gbl}
	g_{\textnormal{BL}}(kh) := \frac{1}{h}\int_{0}^{\infty} g(\tau) \textnormal{ sinc}\left( \frac{kh-\tau}{h}\right) \textnormal{d}\tau, \quad k \in \mathbb{Z}.
	\end{equation}
	Furthermore, the continuous-time frequency response of the system satisfies, for all $|\omega|\leq \omega_B$,
	\begin{equation}
	G(i\omega) = \frac{X(i\omega)}{U(i\omega)} =\frac{X_h(e^{i\omega h})}{U_h(e^{i\omega h})}= h \sum_{k=-\infty}^{\infty} g_{\textnormal{BL}}(kh) e^{-i\omega kh}.  \notag
	\end{equation}
\end{proposition}

Interestingly, by \eqref{gbl} we find that the impulse response $\{g_{\textnormal{BL}}(kh)\}_{k\in\mathbb{Z}}$ is non-causal in general. In other words, a causal continuous-time system behaves like a non-causal system when sampled with a band-limited intersample behavior assumption. Intuitively, this can be deduced by how the intersample behavior of the input is formed: by \eqref{sincinterpolation}, we see that the sinc interpolation of the input must take into consideration the contributions of all the future values of the input at the sampling instants. Thus, the system output will be a function of these future input values as well.

An example of this non-causal behavior can be seen in Figure \ref{fig1}, where we have plotted the impulse response $g(t)$ of a second order continuous-time system and the impulse response of its band-limited discrete-time equivalent $g_{\textnormal{BL}}(kh)$. 	
\begin{figure}
	\centering{
		\includegraphics[width=0.47\textwidth]{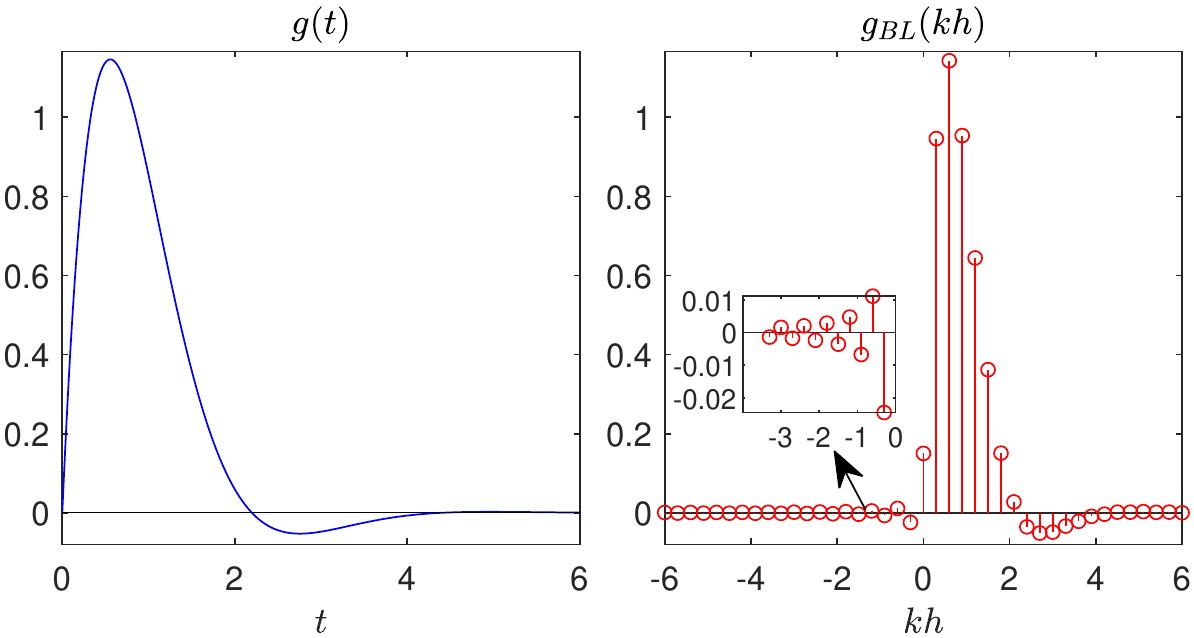}
		\vspace{-0.3cm}
		\caption{Left: Impulse response of a second order continuous-time system. Right: Impulse response of its discrete-time band-limited equivalent.}
		\label{fig1}}
		\vspace{-0.5cm}
\end{figure}
With regards to the behavior of the non-causal part of $\{g_{\textnormal{BL}}(kh)\}_{k\in \mathbb{Z}}$, we note that
\begin{itemize}
	\item
	A significant non-causal part is present if $g(t)$ correlates with $\textnormal{sinc}(t/h-k)$. For example, let $h=1[\textnormal{s}]$ and $g(t)=-e^{-0.2t}\sin(\frac{\pi}{1.1}t)$. As seen in Figure \ref{fig2}, $g(t)$ has an important overlap with the $\textnormal{sinc}$ function, which induces considerable non-causal values in $g_{\textnormal{BL}}(kh)$. Note that such correlation is more likely to occur when the sampling period is close to $\pi/\omega_B$, as in this example.
	\item
	As the sampling period tends to zero, the non-causal part vanishes. In fact, since $h^{-1}\textnormal{sinc}([t-\tau]/h)$ converges weakly \cite{kanwal2011generalized} to $\delta(t-\tau)$, we see that for any fixed $t$ of the form $t=kh$ we have $g_{\textnormal{BL}}(t)\xrightarrow{h\to 0}g(t)$. 
\end{itemize}
\begin{figure}
	\centering{
		\includegraphics[width=0.475\textwidth]{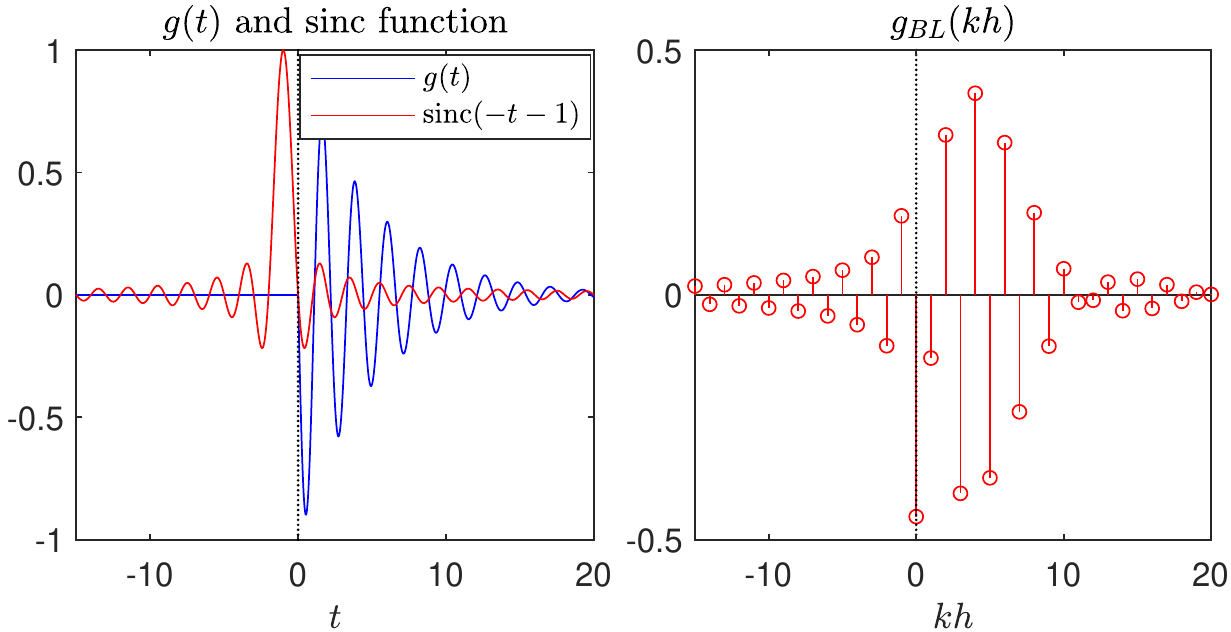}
		\vspace{-0.2cm}
		\caption{Left: Impulse response $g(t)$ with the sinc function that is used for computing $g_{\textnormal{BL}}(-1)$ with \eqref{gbl}. Since the most pronounced lobes are synchronized, the resulting impulse response coefficient is significant. Right: The band-limited equivalent impulse response of $g(t)$.}
		\vspace{-0.2cm}
		\label{fig2}}
\end{figure}		

We now state the problem that is of interest in this paper.
\begin{problem}
	Consider the system described in \eqref{system}, where $\{u(t)\}$ is a band-limited input signal. Assume that we retrieve noisy measurements of the output of the form
	\begin{equation}
	\label{output}
	y(kh) = x(kh) + v(kh), \quad k=1,2,\dots,N,
	\end{equation}
	where $v(kh)$ is a zero-mean stochastic process of variance $\sigma^2$ that is independent of the sampled input. The question we address is how to estimate the continuous-time frequency response $G(i\omega)$ (or equivalently, the band-limited equivalent impulse response $\{g_{\textnormal{BL}}(kh)\}_{k\in\mathbb{Z}}$) from sampled values of the input and noisy output. To avoid further confusion, from now on we denote the true frequency response as $G^*(i\omega)$, and its band-limited equivalent impulse response as $\{g_{\textnormal{BL}}^*(kh)\}_{k\in\mathbb{Z}}$.
	
\end{problem}

\begin{remark}
	It is well known that band-limited signals must extend infinitely in time \cite{lathi2014essentials}. Equivalently, a time-limited signal is not band-limited. Thus, in practice we will encounter \textit{approximately} band-limited signals, which are commonly obtained via anti-aliasing filters \cite{pintelon2012system}. 
\end{remark}
\section{Non-parametric frequency response estimation}
\label{sec:approaches}
In this section we describe our approach for estimating the continuous-time frequency response of a system for band-limited inputs. The sampled output $y(kh)$ can be written as 
\begin{equation}
\label{ykh}
y(kh) = h \sum_{n=-M_{nc}}^{M_c} u([k-n]h)g_{\textnormal{BL}}^*(nh) + w(kh),
\end{equation}
where $M_{nc}$ and $M_{c}$ satisfy $M_{nc}+M_c \geq 0$. These integers indicate the number of non-causal and causal terms of the impulse response that will be estimated. The signal $\{w(kh)\}_{k=1}^N$ is a residual term accounting for the noise sequence $\{v(kh)\}_{k=1}^N$, the approximation error of the series in \eqref{sampledoutput}, and possible transient effects. The equations that the output data satisfies can be put in matrix form as
\begin{equation}
\mathbf{y} = \bm{\Phi}\bm{\rho}^*+\mathbf{w},  \notag
\end{equation}
where
\begin{align}
\mathbf{y} \hspace{-0.06cm}&=\hspace{-0.06cm} \begin{bmatrix}
y(h), \hspace{-0.15cm} & y(2h),  \hspace{-0.2cm} &\dots,  \hspace{-0.15cm} & y(Nh)
\end{bmatrix}^\top, \notag \\
\bm{\Phi} \hspace{-0.06cm}&=\hspace{-0.06cm} h \hspace{-0.12cm}\begin{bmatrix}
u([1\hspace{-0.08cm}+\hspace{-0.08cm}M_{nc}]h\hspace{-0.02cm})\hspace{0.78cm} u(M_{nc}h)  \hspace{0.6cm} \dots \hspace{0.15cm} u([1\hspace{-0.08cm}-\hspace{-0.08cm}M_c]h) \\
u([2\hspace{-0.08cm}+\hspace{-0.08cm}M_{nc}]h) \hspace{0.4cm} u([1\hspace{-0.08cm}+\hspace{-0.08cm}M_{nc}]h) \hspace{0.32cm} \dots \hspace{0.15cm} u([2\hspace{-0.08cm}-\hspace{-0.08cm}M_c]h) \\
\hspace{-0.2cm}\vdots \hspace{2.5cm} \vdots   \hspace{2.7cm}  \vdots   \\
\hspace{-0.04cm}u(\hspace{-0.01cm}[N\hspace{-0.09cm}+\hspace{-0.09cm}M_{nc}]h\hspace{-0.01cm}) \hspace{0.1cm} u(\hspace{-0.01cm}[N\hspace{-0.09cm}-\hspace{-0.09cm}1\hspace{-0.09cm}+\hspace{-0.09cm}M_{nc}]h\hspace{-0.01cm}) \hspace{0.1cm}\dots \hspace{0.1cm} u(\hspace{-0.01cm}[N\hspace{-0.09cm}-\hspace{-0.09cm}M_c]h\hspace{-0.01cm})\hspace{-0.03cm}
\end{bmatrix}\hspace{-0.12cm}, \notag \\
\label{true}
\bm{\rho}^{\hspace{-0.01cm}*} \hspace{-0.06cm}&=\hspace{-0.06cm} \begin{bmatrix}
g_{\textnormal{BL}}^*\hspace{-0.05cm}(\hspace{-0.04cm}-\hspace{-0.02cm}M_{nc}h\hspace{-0.02cm}), \hspace{-0.25cm}&  g_{\textnormal{BL}}^*\hspace{-0.05cm}(\hspace{-0.01cm}[1\hspace{-0.09cm}-\hspace{-0.09cm}M_{nc}]h), \hspace{-0.25cm}& \dots,\hspace{-0.2cm} & \hspace{-0.07cm} g_{\textnormal{BL}}^*\hspace{-0.05cm}(\hspace{-0.02cm}M_c h\hspace{-0.01cm})
\end{bmatrix}^\top\hspace{-0.16cm}, \hspace{-0.2cm}\\
\mathbf{w}\hspace{-0.06cm} &=\hspace{-0.06cm} \begin{bmatrix}
w(h), \hspace{-0.15cm} & w(2h), \hspace{-0.2cm} & \dots,  \hspace{-0.15cm}& w(Nh)
\end{bmatrix}^\top. \notag
\end{align}
As the goal is to provide an estimate for $G^*(i\omega)$, we focus on estimating the vector of coefficients of its truncated Laurent series, $\bm{\rho}^*$. To this end, we will first consider the least-squares estimate of $\bm{\rho}^*$ and study its properties. Afterwards, we present its regularized least-squares variant.

\subsection{Non-causal least-squares estimator}
The least-squares estimate of $\bm{\rho}^*$ is given by
\begin{align}
\hat{\bm{\rho}}_N &= (\bm{\Phi}^\top \bm{\Phi})^{-1} \bm{\Phi}^\top \mathbf{y} \notag \\
\label{lsestimate}
&=\hspace{-0.03cm} \left[\hspace{-0.03cm}\sum_{k=1}^N \bm{\varphi}(kh)\bm{\varphi}^\top(kh)\right]^{\hspace{-0.04cm}-1}\hspace{-0.1cm} \left[\sum_{k=1}^N \bm{\varphi}(kh)y(kh)\hspace{-0.02cm}\right]\hspace{-0.05cm},
\end{align} 
where $\bm{\varphi}\hspace{-0.01cm}(\hspace{-0.01cm}kh\hspace{-0.01cm})$ is the transpose of the $k$-th row of $\bm{\Phi}$. The impulse response estimate generates a non-parametric frequency response estimate $\hat{G}_N(i\omega) = \hat{\bm{\rho}}_N^\top \bm{\Gamma}(e^{i\omega})$, where $\bm{\Gamma}(e^{i\omega})$ is a vector of the form
\begin{equation}
\bm{\Gamma}(e^{i\omega})= h \begin{bmatrix}
e^{i\omega M_{nc} h},  \hspace{-0.2cm} & e^{i\omega (M_{nc}-1) h},  \hspace{-0.2cm}& \dots,  \hspace{-0.2cm} & e^{-i \omega M_c h}
\end{bmatrix}^\top. \notag
\end{equation}
Note that, by construction, the proposed estimate satisfies the conjugacy property
\begin{equation}
\overline{\hat{G}_N(i\omega)} = \hat{\bm{\rho}}_N^\top \overline{\bm{\Gamma}(e^{i\omega})} = \hat{\bm{\rho}}_N^\top \bm{\Gamma}(e^{-i\omega}) = \hat{G}_N(-i\omega). \notag
\end{equation}
To analyze this estimator, we consider inputs of the form
\begin{equation}
\label{inputthm34}
u(t) = \sum_{n=1}^N e(nh) \textnormal{sinc}\left(\frac{t-nh}{h}\right),
\end{equation}
where $\{e(nh)\}_{n=1}^N$ is a white noise sequence of finite variance. For practical purposes, we shall consider that the non-causal samples of the input are all equal to zero (i.e., $u(kh)=0$ for $k<0$). Note that this does not mean that the continuous-time input to the system is causal, as every band-limited signal must extend infinitely in both directions in time.

The following results concern the consistency and asymptotic distribution of the least-squares estimator in \eqref{lsestimate} when the input is discrete-time white noise interpolated through sinc functions. The proofs can be found in the Appendix.

\begin{theorem}
	\label{thm31}
	Consider the system \eqref{system} with measured output \eqref{output}, and $\{u(t)\}$ given by \eqref{inputthm34}. Then, for~any integers $M_{nc}, M_c$ such that $M_{nc}+M_c \geq 0$, we have $\hat{\bm{\rho}}_N \xrightarrow{a.s.} \bm{\rho}^*$, where $\hat{\bm{\rho}}_N$ and $\bm{\rho}^*$ are defined in \eqref{lsestimate} and \eqref{true} respectively.
\end{theorem}

\begin{remark}
	Theorem \ref{thm31} also shows that when causal FIR models are being fit to data with a band-limited input, the coefficients that are estimated converge to the ones provided by the band-limited equivalent, and not the zero-order hold one, which is commonly assumed when discrete-time data is obtained. This fact has implications on the accuracy of the model, as the band-limited equivalent has non-causal coefficients that are different from zero but are usually left unmodeled. As mentioned in Section \ref{sec:preliminaries}, these non-causal terms may only be neglected if the sampling period is small.
\end{remark}

\begin{theorem}
	\label{thm32}
	Consider the system \eqref{system} with measured output \eqref{output}, where the input is given by \eqref{inputthm34}, and $\{e(nh)\}_{n=1}^N$ is white noise of variance $\lambda^2$ that is independent of the output noise sequence $\{v(kh)\}_{k=1}^N$. Then, the least-squares estimate $\hat{\bm{\rho}}_N$ in \eqref{lsestimate} is asymptotically Gaussian distributed, i.e.,
	\begin{equation}
	\sqrt{N}(\hat{\bm{\rho}}_N-\bm{\rho}^*) \xrightarrow{dist.} \mathcal{N}(\mathbf{0},\mathbf{P}_{\textnormal{LS}}), \notag
	\end{equation}
	where the asymptotic covariance matrix is given by
	\begin{equation}
	\mathbf{P}_{\textnormal{LS}}\hspace{-0.05cm}:=\hspace{-0.05cm}\lim_{N\to \infty} \hspace{-0.1cm}\frac{1}{h^4 \lambda^4 N} \hspace{-0.1cm} \sum_{k=1}^N \sum_{n=1}^N \hspace{-0.05cm}\mathbb{E}\{ \hspace{-0.02cm}\bm{\varphi}(kh)w(kh)w(nh) \bm{\varphi}^\top(nh)\hspace{-0.02cm}\}. \notag
	\end{equation}
\end{theorem}

\subsection{Non-causal regularized least-squares estimator}

The proposed least-squares estimator has been shown to be consistent in Theorem \ref{thm31}. However, usually the practitioner is interested in the finite-time behavior, where the number of parameters to be estimated can be of the order of the number of samples. Another situation that may occur is that the number of parameters is larger than the persistence of excitation order of the input signal. In both of these cases, it is convenient to use regularized least-squares estimators.

The regularized least-squares estimate of $\bm{\rho}^*$, denoted here by $\hat{\bm{\rho}}_N^{r},$ is given by
\begin{align}
\hat{\bm{\rho}}_N^{r} &= \arg \min_{\bm{\rho}} \|\mathbf{y}-\bm{\Phi}\bm{\rho}\|_2^2 + \gamma \bm{\rho}^\top \mathbf{P}_r^{-1}\bm{\rho}  \notag \\
\label{regularized}
&= (\mathbf{P}_r\bm{\Phi}^\top \bm{\Phi} + \gamma \mathbf{I}_{M_{nc}+M_c+1})^{-1} \mathbf{P}_r\bm{\Phi}^\top \mathbf{y}, 
\end{align}
where $\mathbf{P}_r\succeq 0$ is a regularization matrix and $\gamma$ is a positive scalar. The problem of choosing the best regularization matrix for causal FIR models has been thoroughly studied during the past years \cite{pillonetto2014kernel,ljung2020shift}. This problem is challenging, since it is known that the optimal regularization matrix depends on the true system \cite{chen2012estimation}. An analogous result holds for non-causal FIR models, as stated next.
\begin{proposition}
	Consider the system with sampled output as in \eqref{ykh}, and assume $\{w(kh)\}$ is white noise of variance $\sigma^2$. The regularization term that minimizes the MSE matrix in~a positive definite sense is given by $\gamma^{\textnormal{opt}}\hspace{-0.05cm}=\hspace{-0.05cm}\sigma^2$ and $\mathbf{P}_r^{\textnormal{opt}}\hspace{-0.04cm} =\hspace{-0.04cm} \bm{\rho}^* \hspace{-0.05cm}{\bm{\rho}^*}^{\hspace{-0.05cm}\top}\hspace{-0.04cm}$, and the corresponding optimal regularized estimate is
	\begin{equation}
	\label{oracle}
	\hat{\bm{\rho}}_N^{r,\textnormal{opt}} \hspace{-0.06cm}=\hspace{-0.06cm} (\bm{\rho}^{\hspace{-0.01cm}*} {\bm{\rho}^{\hspace{-0.01cm}*}}^{\hspace{-0.03cm}\top} \hspace{-0.04cm} \bm{\Phi}^{\hspace{-0.02cm}\top} \bm{\Phi} + \sigma^2 \mathbf{I}_{M_{nc}+M_c+1})^{-1} \bm{\rho}^{\hspace{-0.01cm}*} {\bm{\rho}^{\hspace{-0.01cm}*}}^{\hspace{-0.03cm}\top} \hspace{-0.04cm} \bm{\Phi}^{\hspace{-0.02cm}\top} \mathbf{y}.
	\end{equation}
\end{proposition}
\begin{proof}
	The proof follows by the same reasoning as in the proof of Theorem 1 of \cite{chen2012estimation}.
\end{proof}
Since the optimal regularization matrix is not known a priori, the matrix $\mathbf{P}_r$ is typically parameterized by a low-dimensional hyperparameter vector $\bm{\beta}\in \mathcal{B}$ according to what can be assumed about the impulse response. The difference between the estimation problem in this work and causal FIR estimation is that the impulse response has a causal and non-causal exponential decay, which induces changes in the way kernels should be designed. Recently \cite{blanken2020kernel}, the kernel-design problem for general non-causal systems was studied with the goal of identifying systems with feedforward control. Here we recall the findings in \cite{blanken2020kernel} and apply them to our context. 

Let
\begin{equation}
b_k = \begin{cases}
\lambda_{nc}^{-2k} & \textnormal{if } k<0 \\
\lambda_{c}^{2k} & \textnormal{if } k\geq 0,
\end{cases} \notag
\end{equation}
with $0\leq \lambda_{nc}, \lambda_c<1$. The non-causal tuned/correlated (TC) kernel and the second-order stable spline (SS) kernel yield the following regularization matrices:
\begin{flalign}
&\textit{TC kernel:} \hspace{0.1cm} \mathbf{P}_{\hspace{-0.04cm}r,(j,l)}\hspace{-0.03cm}(\bm{\beta})\hspace{-0.07cm} = \hspace{-0.07cm}\alpha\hspace{-0.03cm} \min\{b_{j-M_{nc}-1}, b_{l-M_{nc}-1}\}; && \notag \\
&\bm{\beta} \hspace{-0.04cm}= \hspace{-0.04cm}[\lambda_{nc},\hspace{0.05cm}\lambda_{c},\hspace{0.05cm} \alpha]^\top; \hspace{0.1cm} \mathcal{B} \hspace{-0.04cm}=\hspace{-0.04cm} \{\bm{\beta}\in\mathbb{R}^3\colon \alpha \hspace{-0.02cm}>\hspace{-0.02cm} 0, \hspace{0.03cm} 0 \hspace{-0.02cm}\leq \hspace{-0.02cm} \lambda_{nc},\hspace{0.02cm}\lambda_{c} \hspace{-0.02cm}< \hspace{-0.02cm}1\}. \notag \\
&\textit{SS kernel:} \hspace{0.1cm} \mathbf{P}_{\hspace{-0.04cm}r,(j,l)}\hspace{-0.03cm}(\bm{\beta})\hspace{-0.08cm} = \hspace{-0.08cm}\frac{\alpha}{6}\hspace{-0.05cm} \min\{b_{j-M_{nc}-1}, b_{l-M_{nc}-1}\}^2 \times && \notag \\
& \left(\hspace{-0.01cm}3\hspace{-0.02cm}\max\{b_{j\hspace{-0.02cm}-\hspace{-0.02cm}M_{nc}\hspace{-0.02cm}-\hspace{-0.02cm}1}, b_{l\hspace{-0.02cm}-\hspace{-0.02cm}M_{nc}\hspace{-0.02cm}-\hspace{-0.02cm}1}\hspace{-0.03cm}\}\hspace{-0.05cm}-\hspace{-0.05cm}\min\{b_{j\hspace{-0.02cm}-\hspace{-0.02cm}M_{nc}\hspace{-0.02cm}-\hspace{-0.02cm}1}, b_{l\hspace{-0.02cm}-\hspace{-0.02cm}M_{nc}\hspace{-0.02cm}-\hspace{-0.02cm}1}\hspace{-0.01cm}\}\hspace{-0.01cm}\right)\hspace{-0.02cm}; && \notag \\
&\bm{\beta} \hspace{-0.04cm}= \hspace{-0.04cm}[\lambda_{nc},\hspace{0.05cm}\lambda_{c},\hspace{0.05cm} \alpha]^\top; \hspace{0.1cm} \mathcal{B} \hspace{-0.04cm}=\hspace{-0.04cm} \{\bm{\beta}\in\mathbb{R}^3\colon \alpha \hspace{-0.02cm}>\hspace{-0.02cm} 0, \hspace{0.03cm} 0 \hspace{-0.02cm}\leq \hspace{-0.02cm} \lambda_{nc},\hspace{0.02cm}\lambda_{c} \hspace{-0.02cm}< \hspace{-0.02cm}1\}.\notag 
\end{flalign}
\begin{remark}
	Using the identities $2\min(a,b)=a+b-|a-b|$ and $2\max(a,b)=a+b+|a-b|$, it can be shown that the kernels above are equivalent to the standard TC and SS causal kernels if $\lambda_{nc}$ is set to zero \cite{blanken2020kernel}.
\end{remark}
To compute the regularized estimator in \eqref{regularized}, all that is left to know is how to tune the hyperparameters in $\bm{\beta}$. This can be done using marginal likelihood optimization with respect to the data, as in the causal estimation case. In other words,
\begin{align}
\hat{\bm{\beta}}_{\textnormal{ML}} &= \arg \max_{\bm{\beta}\in \mathcal{B}} \log(p(\mathbf{y}| \bm{\beta})) \notag \\
&= \arg \min_{\bm{\beta}\in \mathcal{B}} \mathbf{y}^\top \left[\mathbf{Z}(\bm{\beta})\right]^{-1}\mathbf{y} + \log \det(\mathbf{Z}(\bm{\beta})), \notag
\end{align}
where $\mathbf{Z}(\bm{\beta}) := \bm{\Phi}\mathbf{P}_r(\bm{\beta})\bm{\Phi}^\top + \sigma^2 \mathbf{I}_N$. Regarding the variance $\sigma^2$, we provide the computational considerations that must be taken place for including it as a hyperparameter\footnote{As reported in Remark 5 of \cite{pillonetto2014kernel},  $\sigma^2$ can also be estimated separately by computing the sample variance that results from ARX or FIR modeling.}, most of which are included in \cite{chen2013implementation}.

First, we must factor the regularization matrix as $\mathbf{P}_r(\bm{\beta})/\sigma^2 = \mathbf{L(\bm{\beta})L^\top(\bm{\beta})}$. Afterwards, we consider the thin QR factorization \cite[Theorem 2.1.14]{Horn2012}
\begin{equation}
\begin{bmatrix}
\bm{\Phi} \mathbf{L}(\bm{\beta}) & \mathbf{y} \\
\mathbf{I}_{M_{nc}+M_c+1} & \mathbf{0} 
\end{bmatrix} = \mathbf{Q}(\bm{\beta}) \begin{bmatrix}
\mathbf{R}_1(\bm{\beta}) & \mathbf{R}_2 (\bm{\beta})\\
\mathbf{0} & r(\bm{\beta})
\end{bmatrix}, \notag
\end{equation}
where $\mathbf{Q}(\bm{\beta})$ is a rectangular orthogonal matrix, $r(\bm{\beta})$ is a scalar greater than zero, and $\mathbf{R}_1(\bm{\beta})\in \mathbb{R}^{(M_{nc}+M_c+1)\times (M_{nc}+M_c+1)}$ is an upper triangular matrix with positive diagonal entries. Note that the following identities are satisfied:
\begin{align}
\label{qr1}
\hspace{-0.2cm}\mathbf{R}_1^{\hspace{-0.03cm}\top}\hspace{-0.05cm}(\bm{\beta}\hspace{-0.01cm}) \mathbf{R}_1(\bm{\beta}\hspace{-0.01cm}) \hspace{-0.07cm} &= \hspace{-0.07cm} \mathbf{L}^{\hspace{-0.06cm}\top}\hspace{-0.07cm}(\bm{\beta}\hspace{-0.01cm}) \bm{\Phi}^{\hspace{-0.05cm}\top}\hspace{-0.03cm} \bm{\Phi}\mathbf{L}(\bm{\beta}\hspace{-0.01cm})\hspace{-0.07cm} +\hspace{-0.07cm} \mathbf{I}_{M_{nc} \hspace{-0.02cm}+ \hspace{-0.02cm}M_c\hspace{-0.02cm}+\hspace{-0.02cm}1},  \\
\label{qr2}
\hspace{-0.2cm}\mathbf{R}_1^{\hspace{-0.03cm}\top}\hspace{-0.05cm}(\bm{\beta}\hspace{-0.01cm}) \mathbf{R}_2(\bm{\beta}\hspace{-0.01cm}) \hspace{-0.07cm} &=  \hspace{-0.07cm} \mathbf{L}^{\hspace{-0.06cm}\top}\hspace{-0.07cm}(\bm{\beta}\hspace{-0.01cm}) \bm{\Phi}^{\hspace{-0.05cm}\top}\hspace{-0.03cm} \mathbf{y},  \\
\hspace{-0.2cm}\mathbf{R}_2^{\hspace{-0.03cm}\top}\hspace{-0.05cm}(\bm{\beta}\hspace{-0.01cm}) \mathbf{R}_2(\bm{\beta}\hspace{-0.01cm})\hspace{-0.04cm}+\hspace{-0.04cm}r(\bm{\beta})^2 \hspace{-0.07cm} &=\hspace{-0.07cm} \mathbf{y}^\top \hspace{-0.03cm} \mathbf{y}. \notag
\end{align}
After some technical derivations, the equalities above lead to expressing the log-likelihood cost as
\small
\begin{equation}
\mathbf{y}^{\hspace{-0.04cm}\top} \hspace{-0.04cm}[\mathbf{Z}(\bm{\beta})]^{\hspace{-0.03cm}-\hspace{-0.02cm}1}\hspace{-0.04cm}\mathbf{y} + \log\hspace{-0.03cm}\det(\hspace{-0.02cm}\mathbf{Z}(\bm{\beta})\hspace{-0.02cm}) \hspace{-0.095cm} = \hspace{-0.095cm}\frac{r^{\hspace{-0.01cm}2}\hspace{-0.05cm}(\bm{\beta})}{\sigma^2} \hspace{-0.02cm}+ \log(\hspace{-0.015cm}\sigma^{\hspace{-0.01cm}2\hspace{-0.01cm}N}\hspace{-0.045cm})\hspace{-0.01cm}+2\hspace{-0.03cm}\log \hspace{-0.03cm}\det(\hspace{-0.02cm}\mathbf{R}_{\hspace{-0.01cm}1}\hspace{-0.04cm}(\bm{\beta})\hspace{-0.02cm}). \notag
\end{equation}
\normalsize
Since the TC and SS regularization matrices are already factored by a scalar constant $\alpha$, the dependence on $\sigma$ in $\mathbf{L}$ (and therefore in $\mathbf{R}_1$) is redundant for the optimization of the marginal likelihood with respect to $\bm{\beta}$ and $\sigma^2$. Therefore, we can concentrate the cost function by minimizing the log-likelihood cost with respect to $\sigma^2$, which leads to
\begin{equation}
\label{findbeta}
\hat{\bm{\beta}}_{\textnormal{ML}} = \arg \min_{\bm{\beta}\in \mathcal{B}} N\log(r(\bm{\beta})) + \log \det(\mathbf{R}_1(\bm{\beta})).
\end{equation}  
Finally, thanks to \eqref{qr1} and \eqref{qr2}, the regularized least-squares estimate can be computed by
\begin{equation}
\hat{\bm{\rho}}_N^{r} = \mathbf{L}(\hat{\bm{\beta}}_{\textnormal{ML}}) \left[\mathbf{R}_1(\hat{\bm{\beta}}_{\textnormal{ML}})\right]^{-1} \mathbf{R}_2(\hat{\bm{\beta}}_{\textnormal{ML}}). \notag
\end{equation}

\section{Simulations}
\label{sec:simulations}
In this section we illustrate the proposed non-causal estimators though simulation examples, and later verify their advantages through tests on random systems. 

\subsection{Two examples}
\label{sec:simulations_a}
We first consider the following two systems:
\begin{equation}
G_1^*(p) = \frac{1.25}{0.25p^2 + 0.7p+1}, \hspace{0.3cm} G_2^*(p) = \frac{-\pi/1.1}{(p+0.2)^2 + \frac{\pi^2}{1.1^2}},  \notag 
\end{equation}
where $p$ is the differentiation operator, i.e., $px(t) = \frac{dx(t)}{dt}$. These systems have been used for generating the band-limited equivalent impulse responses in Figures \ref{fig1} and \ref{fig2}. The sampling periods are $h_1 = 0.3$[s] and $h_2=1$[s] respectively, and the inputs to $G_1^*(p)$ and $G_2^*(p)$ are given by \eqref{inputthm34}, where $\{e(kh)\}$ is white noise of unit variance. The noiseless continuous-time output is simulated by oversampling the input by a factor of 100, and assuming a first-order hold behavior. White noise is added to the samples of the simulated output, with variance corresponding to an amplitude signal-to-noise ratio of approximately five.

For each system, $N=100$ causal samples are obtained. In order to capture the effect of the non-causal part of the input, we computed the noiseless output $x(t)$ starting from $t=-Nh$, and only the causal sequence $\{x(kh)\}_{k=1}^N$ was contaminated with noise and used for identification. Three causal and four non-causal estimators are tested, all with 40 parameters each; the non-causal ones use $(M_{nc},M_c)=(15,24)$. The following estimators were considered:
\begin{enumerate}
	\item
	Causal least-squares (C-LS);
	\item
	Non-causal least-squares (NC-LS);
	\item
	Causal TC-regularized least-squares (C-TC);
	\item
	Non-causal TC-regularized least-squares (NC-TC);
	\item
	Causal SS-regularized least-squares (C-SS);
	\item
	Non-causal SS-regularized least-squares (NC-SS);
	\item
	Optimal non-causal regularized least-squares (Oracle).
\end{enumerate}
The causal TC and SS-regularized least-squares estimators are obtained via the \texttt{impulseest} command in MATLAB, while the non-causal TC and SS kernels are tuned by solving \eqref{findbeta}. The unrealizable oracle is computed by \eqref{oracle}. The performance of these estimators is compared via Monte Carlo simulations with 300 different noise realizations. Validation data are generated to compute the fit metric
\begin{equation}
\textnormal{Fit} = 100 \left(1- \frac{\|\mathbf{y}-\bm{\Phi}\hat{\bm{\rho}}_N \|_2}{\|\mathbf{y}-\bar{\mathbf{y}}\|_2}\right), \notag
\end{equation}
where $\bar{\mathbf{y}}$ indicates the sample mean of $\mathbf{y}$.

The box plots of the fit of each method, for $G_1^*(p)$ and $G_2^*(p)$, are shown in Figure \ref{fig_fit1}. In both cases there is an advantage in considering non-causal parameters versus fixing them to zero. For $G_1^*(p)$ only a modest improvement can be observed by estimating the non-causal terms, as the sampling period is not large compared to the bandwidth of the system and thus the non-causal component of the impulse response is not significant. On the other hand, the gain in performance in the test with $G_2^*(p)$ is substantial: this is explained by the fact that an important part of the band-limited equivalent impulse response is non-causal. As expected for small sample sizes, an increase in performance is observed on both systems if regularization is included. Note that including regularization in the causal least-squares estimate improves the fit of the causal model but will anyway disregard the significant non-causal components, leading to a worse performance compared to the non-causal regularized least-squares estimators.

\begin{figure}
	\centering{
		\includegraphics[width=0.48\textwidth]{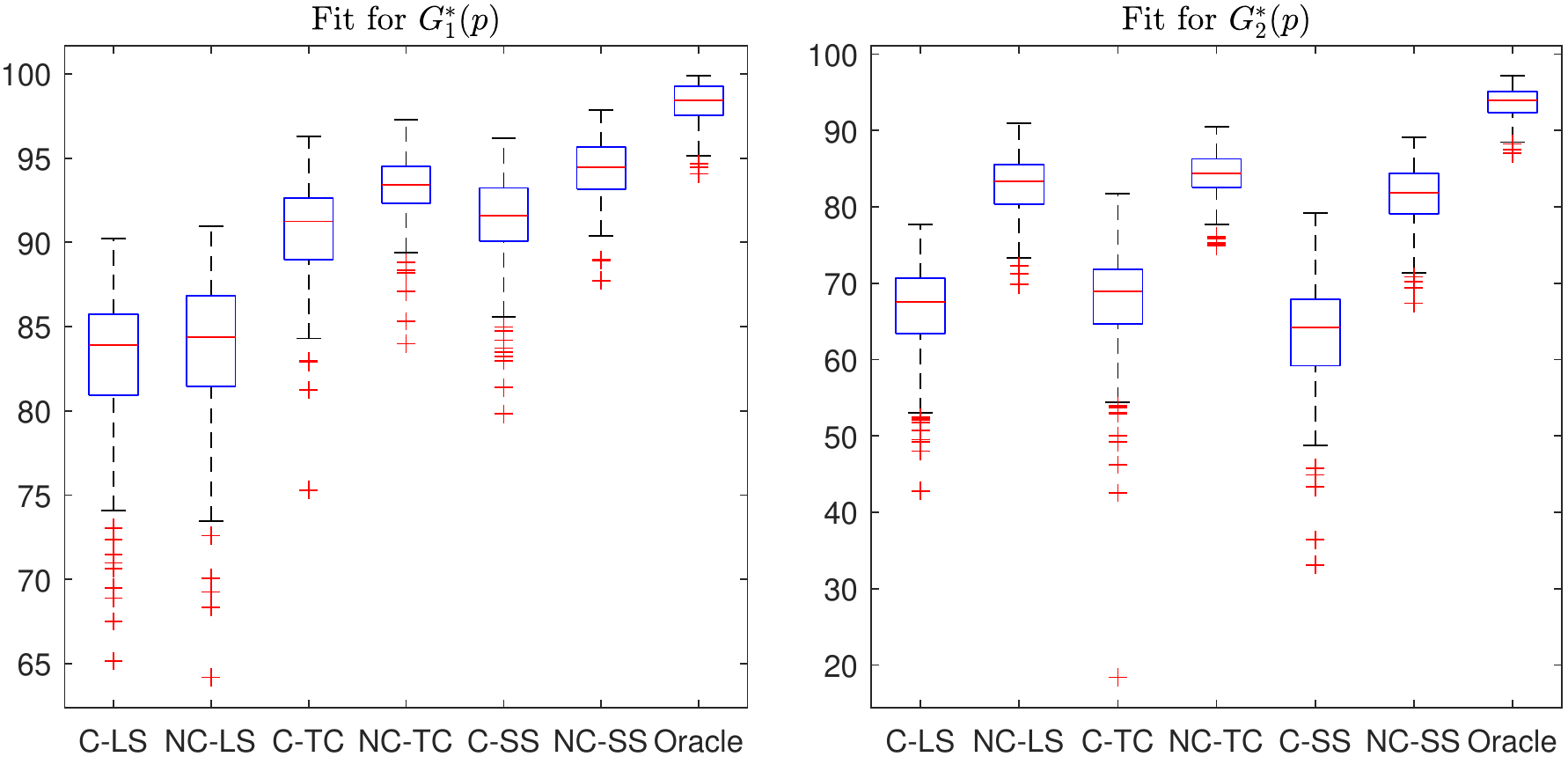}
		\vspace{-0.5cm}
		\caption{Fit box plots of seven different impulse response estimators. Left: $G_1^*(p)$; right: $G_2^*(p)$.}
		\vspace{-0.5cm}
		\label{fig_fit1}}
\end{figure}

\subsection{Random systems}
We now test the proposed methods on a set of random systems. Similar to the bank of test systems presented in \cite{chen2012estimation}, a number of continuous-time systems of order 30 are generated using the \texttt{rss} command in MATLAB, and are sampled at a frequency $f=1/h$ equal to three times the bandwidth. The systems are split into 300 ``fast" systems whose poles have real parts not greater than $\log(0.95)/h$, and 300 ``slow" systems which have at least one pole with real part greater than $\log(0.95)/h$. The systems are excited with the same input as in Section \ref{sec:simulations_a}, and $N=500$ samples are obtained. The outputs are contaminated by Gaussian white noise with an SNR of approximately 20.

All estimators tested previously, except the oracle, are assessed in this new scenario. The box plots of the fit metric for the fast and slow systems are presented in Figure \ref{fig_random_SNR_20}. In both cases, the non-causal estimators are the ones for choice in terms of median fit, which provides strong evidence for the adequacy of including non-causal terms for continuous-time system identification with band-limited input excitation. Note that the ``fast" systems exhibit on average a greater improvement if non-causal terms are included. This can be explained by the fact that the sampling period for this case is relatively large compared to the dominant time-constant, which can induce greater non-causal values for the band-limited equivalent impulse response as studied in Section \ref{sec:preliminaries}.

\begin{figure}
	\centering{
		\includegraphics[width=0.48\textwidth]{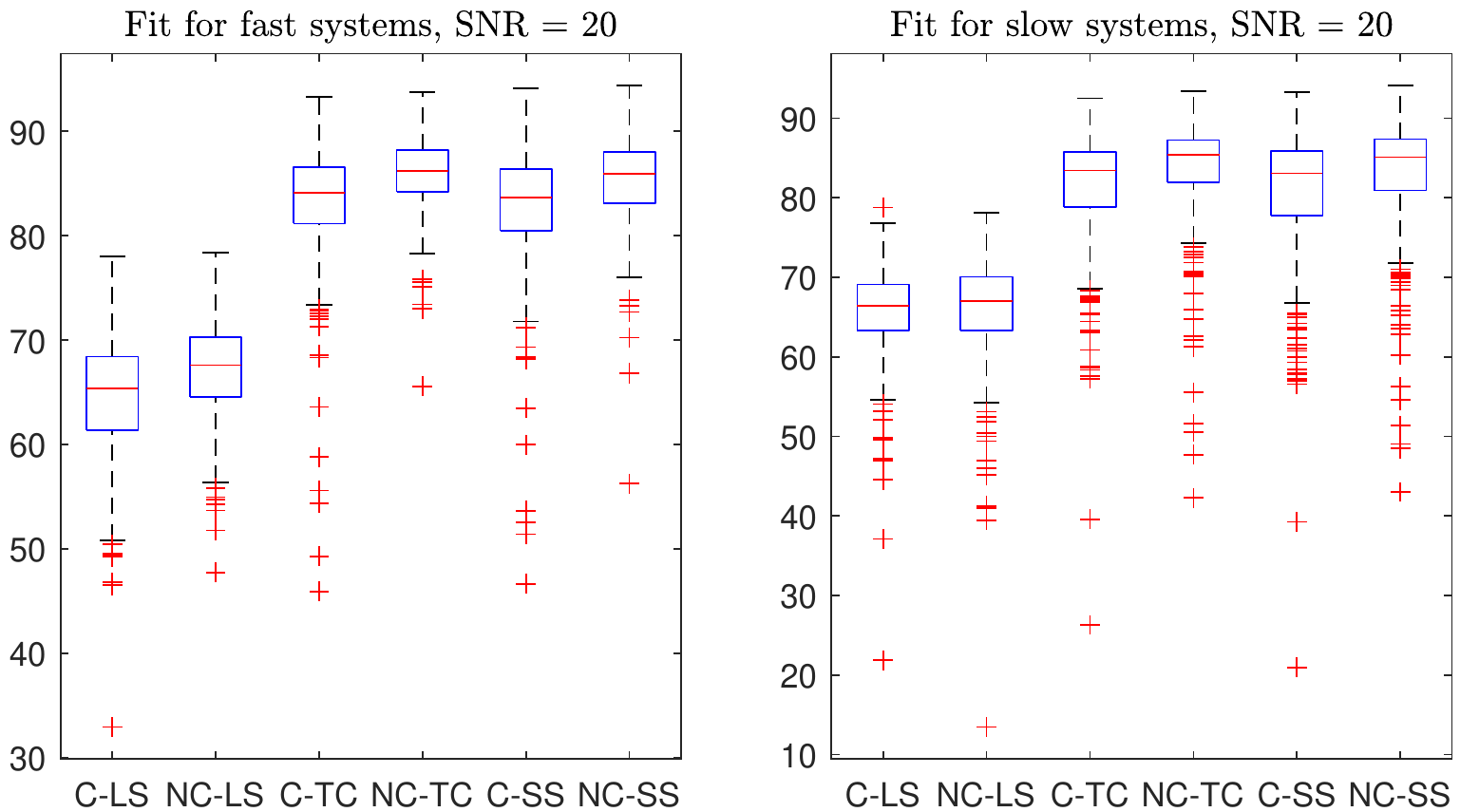}
		\vspace{-0.5cm}
		\caption{Fit box plots of six different impulse response estimators. Left: ``fast" random systems; right: ``slow" random systems.}
		\vspace{-0.4cm}
		\label{fig_random_SNR_20}}
\end{figure}

\section{Conclusions}
\label{sec:conclusions}
In this work we have introduced a novel method for estimating the band-limited equivalent impulse response of a continuous-time system based on non-causal regularized least squares. We began this study by showing that the equivalent discrete-time system for band-limited inputs is non-causal, which led to the analysis of non-causal least-squares estimators for estimating the band-limited equivalent impulse response. The proposed non-causal kernel-based methods show significant advantages in terms of the fit criterion compared to the state-of-the-art methods, since they model the non-causal terms that are commonly overlooked.

\appendix 
\hspace{0.035cm} \textit{Proof of Theorem \ref{thm31}}:
	By leveraging \eqref{ykh} and \eqref{lsestimate}, we have
	\begin{equation}
	\hat{\bm{\rho}}_N \hspace{-0.05cm} = \hspace{-0.05cm} \bm{\rho}^* \hspace{-0.07cm} + \hspace{-0.06cm} \left[\sum_{k=1}^N  \hspace{-0.03cm} \bm{\varphi}(kh)\bm{\varphi}^{\hspace{-0.03cm}\top}\hspace{-0.04cm}(kh)\right]^{ \hspace{-0.08cm}-\hspace{-0.04cm}1} \hspace{-0.05cm}\left[\sum_{k=1}^N \bm{\varphi}(kh)w(kh)\right]\hspace{-0.03cm}, \notag
	\end{equation}
	where $w(kh)$ is given by
	\begin{equation}
	w(kh) = h \hspace{-1cm}\sum_{\{n\colon n<-M_{nc}\}\cup \{n\colon n>M_{c}\}}  \hspace{-1cm} u([k-n]h)g_{\textnormal{BL}}^*(nh) + v(kh). \notag
	\end{equation}
	The ergodic lemma in \cite[Lemma 3.1]{soderstrom1975ergodicity} and the continuous-mapping theorem 
	permit us to write, as $N\to\infty$,
	\begin{equation}
	\label{asymptotics}
	\hat{\bm{\rho}}_N \xrightarrow{a.s.} \bm{\rho}^* + \mathbb{E}\{\bm{\varphi}(kh)\bm{\varphi}^\top(kh)\}^{-1} \mathbb{E}\{\bm{\varphi}(kh)w(kh)\},
	\end{equation}
	where the matrix being inverted is a positive-definite scalar matrix. Since the system in \eqref{system} is asymptotically stable, its band-limited impulse response has finite 2-norm. Thus, thanks to \cite[Lemma 3.1]{soderstrom1975ergodicity} and the fact that $\{u(t)\}$ is white noise at the sampling instants, the following expectation can be computed:
	\begin{flalign}
	&\mathbb{E}\{\bm{\varphi}(kh)w(kh)\} &\notag  \\
	&=h^2 \hspace{-0.55cm}
	\sum_{\substack{\{n: \hspace{0.02cm}n<\hspace{-0.03cm}-\hspace{-0.02cm}M_{nc}\hspace{-0.02cm}\}\cup\\\{n:\hspace{0.02cm} n>\hspace{-0.02cm}M_{c}\}}}
	\hspace{-0.1cm}\begin{bmatrix}
	\mathbb{E}\{\hspace{-0.01cm}u([k\hspace{-0.02cm}-\hspace{-0.02cm}n]h)u([k\hspace{-0.02cm}+\hspace{-0.02cm}M_{nc}]h)\hspace{-0.02cm}\}g_{\textnormal{BL}}^*\hspace{-0.02cm}(nh) \\ 
	\mathbb{E}\{\hspace{-0.01cm}u([k\hspace{-0.02cm}-\hspace{-0.02cm}n]h)u([k\hspace{-0.02cm}-\hspace{-0.02cm}1\hspace{-0.02cm}+\hspace{-0.02cm}M_{nc}]h)\hspace{-0.02cm}\}g_{\textnormal{BL}}^*\hspace{-0.02cm}(nh) \\  
	\vdots \\  
	\mathbb{E}\{\hspace{-0.01cm}u([k\hspace{-0.02cm}-\hspace{-0.02cm}n]h)u([k\hspace{-0.02cm}-\hspace{-0.02cm}M_c]h)\}g_{\textnormal{BL}}^*\hspace{-0.02cm}(nh)
	\end{bmatrix} & \notag \\
	&= \mathbf{0}. &\notag
	\end{flalign}
	This result, together with \eqref{asymptotics}, leads to the desired conclusion. \hfill $\square$ 

\hspace{0.035cm} \textit{Proof of Theorem \ref{thm32}}:
	We have
	\begin{equation}
	\sqrt{ \hspace{-0.05cm}N} \hspace{-0.02cm}( \hspace{-0.01cm}\hat{\bm{\rho}}_{\hspace{-0.02cm}N} \hspace{-0.02cm}-\hspace{-0.02cm}\bm{\rho}^{\hspace{-0.01cm}*} \hspace{-0.03cm}) \hspace{-0.1cm}= \hspace{-0.12cm} \left[ \hspace{-0.08cm}\frac{1}{N} \hspace{-0.07cm}\sum_{k=1}^N \hspace{-0.07cm} \bm{\varphi}(\hspace{-0.01cm}kh \hspace{-0.02cm})\bm{\varphi}^{\hspace{-0.05cm}\top} \hspace{-0.06cm}(\hspace{-0.01cm}kh\hspace{-0.02cm}) \hspace{-0.04cm}\right]^{\hspace{-0.1cm}- \hspace{-0.03cm}1} \hspace{-0.12cm}\left[ \hspace{-0.08cm}\frac{1}{\sqrt{\hspace{-0.05cm}N}} \hspace{-0.07cm}\sum_{k=1}^N  \hspace{-0.07cm}\bm{\varphi}(\hspace{-0.01cm}kh \hspace{-0.02cm})w(\hspace{-0.01cm}kh\hspace{-0.02cm})\hspace{-0.04cm}\right] \hspace{-0.1cm}. \notag
	\end{equation}
	Since the first sum converges to its expected value for large $N$, we can write
	\begin{equation}
	\sqrt{N}(\hat{\bm{\rho}}_N-\bm{\rho}^*) = \frac{1}{h^2 \lambda^2} \left[\frac{1}{\sqrt{N}}\sum_{k=1}^N \bm{\varphi}(kh)w(kh)\right] + o_p(1). \notag
	\end{equation}
	By Lemma A4.1 of \cite{soderstrom1983instrumental}, the sum in brackets above converges in distribution to a zero-mean normal random variable with covariance
	\begin{equation}
	\mathbf{P}= \lim_{N\to \infty} \frac{1}{N} \sum_{k=1}^N \sum_{n=1}^N \mathbb{E}\{ \bm{\varphi}(kh)w(kh)w(nh) \bm{\varphi}^\top(nh)\}. \notag
	\end{equation} 
	The asymptotic covariance $\mathbf{P}_{\textnormal{LS}}$ of Theorem \ref{thm32} is thus~obtained by applying Lemma A4.2 of \cite{soderstrom1983instrumental} and its corollary.~$\square$ 

\bibliography{References}

\begin{thebibliography}{10}

\bibitem{rao2006identification}
G.~Rao and H.~Unbehauen, ``Identification of continuous-time systems,'' {\em
  IEE Proceedings-Control theory and applications}, vol.~153, no.~2,
  pp.~185--220, 2006.

\bibitem{garnier2008book}
H.~Garnier and L.~Wang, eds., {\em Identification of Continuous-time Models
  from Sampled Data}, Springer, 2008.

\bibitem{pintelon1997frequency}
R.~Pintelon, J.~Schoukens, and G.~Vandersteen, ``Frequency domain system
  identification using arbitrary signals,'' {\em IEEE Transactions on Automatic
  Control}, vol.~42, no.~12, pp.~1717--1720, 1997.

\bibitem{levy1959complex}
E.~Levy, ``Complex-curve fitting,'' {\em IRE transactions on automatic
  control}, no.~1, pp.~37--43, 1959.

\bibitem{gonzalez2020consistent}
R.~A. Gonz{\'a}lez, C.~R. Rojas, S.~Pan, and J.~S. Welsh, ``Consistent
  identification of continuous-time systems under multisine input signal
  excitation,'' {\em \textnormal{Submitted for publication to} Automatica},
  2020.

\bibitem{feuer1996sampling}
A.~Feuer and G.~Goodwin, {\em Sampling in {D}igital {S}ignal {P}rocessing and
  {C}ontrol}.
\newblock Birkh\"auser, 1996.

\bibitem{relan2016recursive}
R.~Relan and J.~Schoukens, ``Recursive discrete-time models for continuous-time
  systems under band-limited assumptions,'' {\em IEEE Transactions on
  Instrumentation and Measurement}, vol.~65, no.~3, pp.~713--723, 2016.

\bibitem{rabiner1978fir}
L.~Rabiner, R.~Crochiere, and J.~Allen, ``{FIR} system modeling and
  identification in the presence of noise and with band-limited inputs,'' {\em
  IEEE Transactions on Acoustics, Speech, and Signal Processing}, vol.~26,
  no.~4, pp.~319--333, 1978.

\bibitem{lu2019identification}
Q.~Lu, P.~D. Loewen, R.~B. Gopaluni, M.~G. Forbes, J.~U. Backstr{\"o}m, G.~A.
  Dumont, and M.~S. Davies, ``Identification of symmetric noncausal
  processes,'' {\em Automatica}, vol.~103, pp.~515--530, 2019.

\bibitem{forssell2000projection}
U.~Forssell and L.~Ljung, ``A projection method for closed-loop
  identification,'' {\em IEEE Transactions on Automatic Control}, vol.~45,
  no.~11, pp.~2101--2106, 2000.

\bibitem{aljanaideh2017closed}
K.~F. Aljanaideh and D.~S. Bernstein, ``Closed-loop identification of unstable
  systems using noncausal {FIR} models,'' {\em International Journal of
  Control}, vol.~90, no.~2, pp.~168--185, 2017.

\bibitem{boas1972summation}
R.~Boas~Jr., ``Summation formulas and band-limited signals,'' {\em Tohoku
  Mathematical Journal, Second Series}, vol.~24, no.~2, pp.~121--125, 1972.

\bibitem{kanwal2011generalized}
R.~P. Kanwal, {\em Generalized {F}unctions: {T}heory and {A}pplications,
  \textnormal{3rd Edition}}.
\newblock Springer, 2011.

\bibitem{lathi2014essentials}
B.~P. Lathi and R.~A. Green, {\em Essentials of digital signal processing}.
\newblock Cambridge University Press, 2014.

\bibitem{pintelon2012system}
R.~Pintelon and J.~Schoukens, {\em System {I}dentification: {A} {F}requency
  {D}omain {A}pproach}.
\newblock John Wiley \& Sons, 2012.

\bibitem{pillonetto2014kernel}
G.~Pillonetto, F.~Dinuzzo, T.~Chen, G.~De~Nicolao, and L.~Ljung, ``Kernel
  methods in system identification, machine learning and function estimation:
  {A} survey,'' {\em Automatica}, vol.~50, no.~3, pp.~657--682, 2014.

\bibitem{ljung2020shift}
L.~Ljung, T.~Chen, and B.~Mu, ``A shift in paradigm for system
  identification,'' {\em International Journal of Control}, vol.~93, no.~2,
  pp.~173--180, 2020.

\bibitem{chen2012estimation}
T.~Chen, H.~Ohlsson, and L.~Ljung, ``On the estimation of transfer functions,
  regularizations and {G}aussian processes--{R}evisited,'' {\em Automatica},
  vol.~48, no.~8, pp.~1525--1535, 2012.

\bibitem{blanken2020kernel}
L.~Blanken and T.~Oomen, ``Kernel-based identification of non-causal systems
  with application to inverse model control,'' {\em Automatica}, vol.~114,
  p.~108830, 2020.

\bibitem{chen2013implementation}
T.~Chen and L.~Ljung, ``Implementation of algorithms for tuning parameters in
  regularized least squares problems in system identification,'' {\em
  Automatica}, vol.~49, no.~7, pp.~2213--2220, 2013.

\bibitem{Horn2012}
R.~A. Horn and C.~R. Johnson, {\em Matrix Analysis, \textnormal{2nd Edition}}.
\newblock Cambridge University Press, 2012.

\bibitem{soderstrom1975ergodicity}
T.~S{\"o}derstr{\"o}m, ``Ergodicity results for sample covariances,'' {\em
  Problems of Control and Information Theory}, vol.~4, no.~2, pp.~131--138,
  1975.

\bibitem{soderstrom1983instrumental}
T.~S{\"o}derstr{\"o}m and P.~Stoica, {\em Instrumental {V}ariable {M}ethods for
  {S}ystem {I}dentification}.
\newblock Springer, 1983.

\end{thebibliography}
\end{document}